\documentclass[11pt]{article}

\usepackage{booktabs} 

\usepackage{wrapfig}
\usepackage{tikz}  
\usetikzlibrary{positioning}
\usetikzlibrary{decorations.pathmorphing,shapes}
\usepackage{amsmath,amssymb}
\newcommand{\semb}{[ \! [}
\newcommand{\seme}{] \! ]}
 \newcommand\ForAuthors[1]
 {\par\smallskip                     
  \begin{center}
   \fbox
   {\parbox{0.9\linewidth}
    {\raggedright\sc--- #1}
   }
  \end{center}
  \par\smallskip                     %
 }        

\newcommand{\footremember}[2]{%
    \footnote{#2}
    \newcounter{#1}
    \setcounter{#1}{\value{footnote}}%
}
 
\newtheorem{definition}{Definition}
\newtheorem{lemma}{Lemma}
\newtheorem{theorem}{Theorem}

\newtheorem{corollary}{Corollary}
\newtheorem{proof}{Proof}

\def\cG{{\mathcal G}}
\def\cI{{\mathcal I}}
\def\cK{{\mathcal K}}

\def\cP{{\mathcal P}}

\def\cR{{\mathcal R}}

\def\cT{{\mathcal T}}

\def\Gz{\cG_0}

\def\Gp{\cG_P}
\def\Gout{\cG_{out}}
\def\Gfull{\cG_{full}}

\newcommand{\I}[1]{\mathit{#1}}
\DeclareMathOperator{\Lit}{Lit}
\newcommand{\Msf}{\textsf{M}}
\newcommand{\Ssf}{\textsf{S}}
\newcommand{\ssf}{\textsf{s}}
\newcommand{\pre}{\textsf{pre}}
\newcommand{\Lcal}{\mathcal{L}}

\newcommand{\la}{\langle}
\newcommand{\ra}{\rangle}

\newcommand{\sched}{{\sf Sched}}

\newcommand{\commentA}[1]{}

\newcommand{\RM}[1]{{\rm{#1}}}
\newcommand{\tsf}{\textsf{t}}

\begin{document}
\title{Models of fault-tolerant distributed computation via dynamic epistemic logic}
\author{Eric Goubault\footremember{X}{LIX, Ecole Polytechnique, CNRS, Universit\'e Paris-Saclay, 91128 Palaiseau, France \texttt{goubault@lix.polytechnique.fr}} \and 
Sergio Rajsbaum\footremember{UNAM}{Instituto de Matematicas, UNAM, Ciudad Universitaria Mexico 04510, Mexico \texttt{rajsbaum@im.unam.mx}}}

\maketitle

\begin{abstract}
  The computability power of a \emph{distributed computing model} is determined by the communication media available to the processes,
  the timing assumptions about processes and communication, and the nature of failures that processes  can suffer. 
     In a companion paper we showed how dynamic epistemic logic can be used to give a formal semantics
     to a given distributed computing model, to capture precisely the knowledge needed to solve a distributed \emph{task},
     such as consensus.
     Furthermore, by moving to a dual model
     of epistemic logic defined by simplicial complexes,  topological invariants are exposed, which determine task solvability.
       In this paper we show how to extend the setting above to include in the knowledge of the processes,  knowledge about the model of computation itself.
       The extension describes  the knowledge processes gain  about the current execution,  in problems
       where processes  have no input values at all. 
\end{abstract}
%


\maketitle

\newpage

\tableofcontents

\newpage


\section{Introduction}

Dynamic epistemic logic (DEL) considers multi-agents systems and studies 
how knowledge changes when communication events occur.
An epistemic  S5 model is typically used to represent  states of a multi-agent system, where edges 
 of the Kripke structure are labeled
with the agents that do not distinguish between the two states.
A Kripke model represents the knowledge of the agents about an initial situation,
and an \emph{action model} represents their knowledge about the possible events taking place in this situation.
A product update operator  defines the Kripke model that results as a 
  consequence of executing actions on the initial  model.
   In the simplest case, public announcement to all the agents of a formula $\psi$ are considered,
   but there is a  general logical language to reason about information and knowledge change~\cite{baltag&:98,BaltagM2004}
   to represent the execution of actions that are indistinguishable to a process.
 
  We are interested in using DEL to study the computability power of a \emph{distributed computing model}.
  It is known that the computability power of a model is determined by the communication media available to the processes,
  the timing assumptions about processes and communication, and the types of failures that processes 
  can suffer.
  The basic model consists of a set of processes communicating by writing and reading shared registers, each process runs at
  its own speed that can vary and is independent of other processes speed, and any number of processes can fail by crashing.
   A task, such as \emph{consensus}, 
     is defined by possible input values to the processes, 
     output values to be produced at the end of the protocol, and an input/output relation. 
 The \emph{wait-free theorem} of~\cite{1999TopologicalStructureAsynchronous_HS} characterizes the \emph{tasks} that are solvable in this model, by
exposing the intimate relation between topology and distributed computing.
It shows that the topology of the input complex is fully preserved after a read/write wait-free protocol is executed, and paved the way
     to show that other models also carry topological information that determines their computability power; for an overview of the theory see~\cite{HerlihyKR:2013}.

  In a companion paper~\cite{ericSergioDEL1-2017} we show how DEL can be used to give a formal semantics
     to a given distributed computing model,  capturing precisely the knowledge needed to solve a {task}. 
%
To expose the  underlying topological invariants induced by the action model,
a \emph{simplicial complex} model corresponding,
in a precise categorical sense, to the dual of the Kripke structure is used. In the figure below, $I$ is the \emph{input model}, an initial
epistemic simplicial complex model (equivalent to Kripke model), and the \emph{protocol model} $P$ is the product with an action model,
representing the knowledge gained after a certain number of communication steps.
The action model preserves topological invariants from the initial model $I$  to the complex $P$ after
the communication actions have taken place.
We explored a class of action models that \emph{fully} preserve the topology of the initial complex.
\begin{wrapfigure}[9]{r}{-0.4\textwidth}
 \centering
\begin{tikzpicture}
  \node (s) {$P$};
  \node (xy) [below=2 of s] {$\Delta \subseteq I \times O$};
  \node (x) [left=of xy] {$I$};
  \draw[<-] (x) to node [sloped, above] {$\pi_I$} (s);
  \draw[->, dashed, left] (s) to node {$h$} (xy);
  \draw[->] (xy) to node [below] {$\pi_I$} (x);
\end{tikzpicture}
\end{wrapfigure}
For a given task, we  defined another \emph{knowledge goal} action model, that when  used
to make the product with the initial epistemic model, yields an epistemic model $\Delta$
representing what the agents should be able to know to solve the task, after applying the communication action model.
There is sufficient knowledge in $P$ if there exists a (properly defined)  morphism $h$ from $P$ to  $\Delta$ that makes the diagram (of underlying Kripke frames) commute.

\paragraph*{Motivation}
While many distributed problems have  the flavor of a distributed function, and can be defined as a task,
some  actually do not refer to input/output relations, but to properties about the execution
itself;  processes have no inputs at all.
All these problems share a striking commonality.
There is just one initial state of the system, and the initial Kripke model of the setting in~\cite{ericSergioDEL1-2017} consists of just one state.
Indeed, all processes have exactly the same knowledge initially.
Then, no matter what the action model is, in the final protocol Kripke model after any number of communication steps, 
the processes gain no knowledge. 
This indeed implies that there is nothing the processes can compute in terms of producing outputs from input values.
This  contradicts the fact that there are wait-free distributed algorithms for many inputless problems.

\paragraph*{Contributions}
Indeed, the problem above is that the processes do not know the model of computation itself.
We propose a novel way of using DEL to encode a given model of computation as knowledge that the processes have
in the initial Kripke model. The processes know the model in the sense that they know which actions the environment
can take, and their structure. 
Then, in the protocol Kripke model, processes do gain knowledge about the execution. 
How much knowledge is determined by the model of computation itself, encoded in the action model.
Then we can indeed prove, that an {inputless problem} is solvable if and only if in the protocol Kripke model
the processes gained sufficient knowledge to solve it.

In this paper we work out only the case of  \emph{inputless tasks,} where each process should produce a single output value.
We first present the setting using Kripke models, and then we observe that by moving from Kripke models to their dual as in ~\cite{ericSergioDEL1-2017},  
topological invariants are exposed which determine the solvability of a given inputless task.
Also, we concentrate on iterated models, where the setting becomes very elegant, due to their recursive nature.
Furthermore, for concreteness, we work with the case where processes communicate by  a sequence of shared arrays~: 
they all go through the same sequence in the same order, i.e. in each one, they first write a value and then they take an
atomic snapshot of the array. 

In this setting our action model becomes very simple. The environment can schedule the processes
to do their operations in a given round, by deciding an  interleaving of their write and read operations.
We call such an interleaving a joint action of the environment.
To be more specific, the initial Kripke model (that we call $M_0$ in the sequel) 
contains one state for each such joint action. The atomic proposition associated 
with each of these states 
means, informally, that the processes consider the joint action as a possible future event of the environment.
Furthermore, in the input model $M_0$ all states are indistinguishable to all processes.
Now, the action model $A$ has also one action point for each possible joint action, but here the accessibility relation
is non-trivial. Two states are related to process $i$ if the process could not distinguish which of the two joint actions
actually took place. Then, the product of $M_0$ and $A$ gives a protocol model $M_1$, which then should
have sufficient knowledge to solve the task.


\paragraph*{Related work}
Seminal work on knowledge and distributed systems is of course one of the inspirations
of the present work (and of~\cite{ericSergioDEL1-2017}), e.g. \cite{FHMVbook}, as well as the  combinatorial topology approach for fault-tolerant distributed computing,
see e.g. \cite{HerlihyKR:2013}. 
But the authors know no previous work on relating the combinatorial
topological methods of \cite{HerlihyKR:2013} with Kripke models.
It should be mentioned though that between Kripke models and interpreted systems have also been compared, from a categorical perspective in e.g.~\cite{PORTER2004235}.
 
In this paper  we use dynamic epistemic logic (DEL)~\cite{sep-dynamic-epistemic,DEL:2007}.
Complex epistemic actions can be
represented in \emph{action model logic} \cite{BaltagM2004,DEL:2007}. Various examples of epistemic actions have been considered, especially 
\emph{public announcement logic},  a well-studied example of DEL, with many applications in dynamic logics, knowledge representation and other
formal methods areas. However, to the best of our knowledge, it has not been used in distributed computing theory, where fault-tolerance
is of primal interest. 
DEL~\cite{BaltagM2004,DEL:2007}
extends epistemic logic through dynamic operators formalizing
information change.
Plaza~\cite{plaza:89} first extended epistemic logic to model public
announcements, where the same information is transmitted to all
agents.
Next, a variety of approaches (e.g.,~\cite{baltag&:98,DEL:2007}) generalized
such a logic to include communication that does not necessarily reach
all agents.
Here, we build upon the approach developed by Baltag et
al.~\cite{baltag&:98} employing action models.
We have focused in this paper on the classical semantics of multi-modal S5 logics. 
%
 
Many inputless problems have been considered in the past.
For  example, the \emph{participating set} problem~\cite{BorowskyG1993} and its variants, where
 processes should produce as output sets
of processes ids that they have seen participating in an execution, such that any two such sets can be ordered by containment,
plays a role in the set agreement impossibility proof of~\cite{1993GeneralizedFLPImposibility_BG,SaksZ00}.
Other examples  of inputless problems include 
 the timestamp object of~\cite{EFR08} (called weak counter in \cite{GR08}).
A weak counter provides a single operation, \texttt{Get-Timestamp}, which returns an integer. It has the property that if one operation 
precedes another, the value returned by the later operation must be larger than the
value returned by the earlier one. (Two concurrent \texttt{Get-Timestamp} operations may return the same value.)
 Also, the \texttt{test\&set} object,
specifying that in any execution exactly one process should output $1$ and the others should output $0$.

\section{Distributed systems background}
We describe our setting in a concrete family of models, that are of interest in distributed computing. 
We first recall some basic  notions about shared memory computation, e.g.~\cite{2004dc_AW,HSbook}.
Then we represent the executions of a model in a  state/transition framework adapting
 adapt  the model of \cite{MosesR2002} (in turn following  the style of \cite{FHMVbook}).
 
 \subsection{Distributed computing models}
 \label{subsec:model}
 
 Our basic model  is the   {\it one-round read/write} asynchronous model, $\mathsf{WR}$.
It  consists of $n+1$ processes  denoted by the numbers $[n]=\{0,1,\ldots,n\}$, referred to as \emph{ids}.
 A process is a deterministic (possibly infinite) state machine. 
Processes  communicate through a shared memory array $\mathsf{mem}[0\ldots n]$ which consists of $n+1$ single-writer/multi-reader atomic registers. Each process  accesses the shared memory by invoking the atomic operations $\mathsf{write}(x)$ or $\mathsf{read}(j)$, $0\leq j\leq n$. The $\mathsf{write}(x)$ operation is used by process $i$  to write  value $x$ to its own register, $i$, and process $i$  can invoke $\mathsf{read}(j)$ to read register $\mathsf{mem}[j]$, for any $0\leq j\leq n$. 
 Any interleaving of the $\mathsf{write}()$ and $\mathsf{read}(j)$ operations of the processes is  possible.
The protocol $D$ that the processes execute represent their state machines, they 
define the next operation to execute, and what to remember.  To have concrete examples, it is convenient
to assume the protocol has the following canonical form.
 In its first operation, process $i$ writes a value  to $\mathsf{mem}[i]$, then 
 a process  reads each of the $n+1$ registers, in an arbitrary order. 
 Such a sequence of  read operations, is abbreviated by  $\mathsf{Collect}()$, and when it is preceded by a  $\mathsf{write}(x)$
 it is  abbreviated by  $\mathsf{WCollect}(x)$. 
In the  { one-round read/write} asynchronous model, $\mathsf{WR}$, the protocol of each process consists
of a single  $\mathsf{WCollect}(x)$.
A protocol in this canonical form has to determine only 
the values the processes write to the shared memory,
and what do they remember about the values read from the shared memory, but the next operation to execute
is determined by the round structure of the $\mathsf{WR}$ model.
 More generally, in the $N$-{\it multi-round read/write} model $\mathsf{WR}$, the program of every process consists
of a sequence of $N$  $\mathsf{WCollect}()$  operations.


The \emph{iterated} $\mathsf{WR}$ model, is obtained by composing the one-round $\mathsf{WR}$ model $N$ times.   
Processes communicate through a sequence of arrays, $\mathsf{mem}_1$, $\mathsf{mem}_2\ldots,\mathsf{mem}_N$.
 They all go through the sequence of arrays, executing a
single $\mathsf{WCollect}()$ operation on $\mathsf{mem}_r$, for each $r\geq 0$.
Namely,  each process $i$ executes one write to $\mathsf{mem}_r[i]$  and then reads one by one all entries $j$, $\mathsf{mem}_r[j]$, 
in arbitrary order, before proceeding to do the same on $\mathsf{mem}_{r+1}$.
Again, any interleaving of the operations of the processes is possible.

Several sub-models have been considered in the literature, equivalent to each other in terms of their task computability power (for an overview of such results see~\cite{HerlihyKR:2013}).
The \emph{snapshot} version of the previous  models, is obtained by replacing
 the $\mathsf{WCollect}()$ by a $\mathsf{WSnap}()$ operation, that guarantees that the reads 
 of  the $n+1$ registers happen all atomically, at the same time. To obtain  versions that tolerate $t$ processes crashing,
 in each round, a process writes a value and then repeatedly reads the shared memory (either using a collect or snapshot, depending on the model)  until it sees that at least
 $n+1-t$ processes have written a value for that round.
Finally, the more structured \emph{immediate snapshot} $\mathsf{WR}$ models,
are such that executions are organised in concurrency classes (here a $t$-resilient version is not obtained directly~\cite{tResIS-DFRR2016}). 
Each concurrency class consists of a set of processes,
that are all scheduled to do their write operations concurrently, and then they all execute a snapshot operation concurrently.

 In this paper we assume processes have no inputs; all processes are in the same initial state (differing only in their process ids).

 \subsection{States and actions}
 \label{sec:statActs}

 In addition to the  set of processes,  $0,1,2,\ldots, n$, there is 
an {\em environment,} denoted by~$e$, which is
used to model the  shared memory, as well as the scheduling of the operations of the processes.
For every  $i\in\{ e,0,1,\ldots,n\}$, there is
a set~$L_i$ consisting of all possible {\em local states} for~$i$.
The set of {\em global states,}  simply called {\em states},
 consist of
$
\cG\ = \ L_e\times L_0\times\cdots\times L_n.
$
We denote by $x_i$ the local state of~$i$ in the state~$x$.
Notice that in the above models, given a state $x$, that includes the contents of the shared-memory in the environment's state,
scheduling a set of processes to execute their next operations, uniquely determines the subsequent state of the system~\footnote{
When  non commuting operations are included in the schedule, we assume some fixed a priori ordering to execute them, say always first
the writes and then the reads.}.
This is the state resulting by executing 
the next operations by  the scheduled processes,  and updating the environment's state to reflect the contents of the shared-memory accordingly,
and the new local states of the processes that executed an operation.
A {\em scheduling action} is a set $\sched\subseteq \{0,\ldots,n\}$
of the processes that are scheduled to move next.
Thus, a \emph{run} of a deterministic protocol~$D$ can be represented in the form
$x\odot sa_1\odot sa_2\odot\cdots$ where $x$ is
an initial state and   $sa_i$ is a scheduling action
for every integer $i\ge 0$.
  An {\em execution} is a  subinterval of a run,
starting and ending in a state.
Sometimes it is convenient to talk about \emph{composed scheduling actions}, consisting
of a sequence of scheduling actions $[ sa_1,sa_2,\ldots,sa_k]$.
Given an execution $R$ (possibly consisting of just one state), and a (possibly composed) scheduling action $sa$,
 $R\odot sa$ denotes the execution that results from
extending $R$ by  performing the (composed) scheduling action~$sa$.


Let $\mathsf{Sch}$ be the set of all infinite schedules of a given distributed computing model,
and let $N-\mathsf{Sch}$ the set of all the prefixes, where every process takes $N$ steps.
The set of all $N$-step schedules for the write-collect models is denoted $\mathsf{WR}_R$, and can be viewed
as  all permutations of the process ids, where each id appears exactly $N$ times (although we
often organise such a permutation as a sequence of sets, a composed scheduling action).
The set of all $N$-step schedules for the other models are denoted analogously, $\mathsf{IS}_R$
and  $\mathsf{IIS}_R$. Thus, when $N$ is clear from the context, we omit it.

Finally, we   define an $N$-step  \emph{model aware system} that will serve us well  to define an epistemic model, by including in the environment's state
 the schedule itself. This implies that the environment also runs a deterministic protocol.
 Namely, for each possible $N$-step schedule in the model, there is one initial state of the environment. 
 All initial states of the environment have the shared memory empty. There is a single initial state of each process $i$, associated to its id $i$. Thus, the \emph{initial states}~$\Gz$ of a model are as follows. 
 The initial states of all processes are identical, except that the initial state of process $i$
contains its id $i$. The initial states of the environment are in a 1 to 1 correspondence to all possible 
composed scheduling actions of the model.
For each $N$-step composed scheduling 
action   there
is an initial state of the environment. In addition, the environment's state encodes that the shared memory is initially empty.
Thus, each state in  $\Gz$ can be denoted as $x=([s_0,\ldots,s_k], q_0,q_1,\ldots,q_n)$,
where $x_e=[s_0,\ldots,s_k]$ is an initial state of the environment (specifying that each process takes $N$ steps,
and in which order) and $x_i=q_i$ is the initial state of process $i$. 
The only scheduling action that can be applied to $x$ is $[s_0,\ldots,s_k]$, and once a distributed protocol $D$
is fixed, it defines a unique
 execution,  $x\odot [s_0,\ldots,s_k]$ whose last state is a \emph{$N$-step protocol state}. The set of all such
 states is denoted $\Gp$.

 \section{Distributed computing with DEL}
\label{sec:modelDEL}
The distributed computing modeling in Section~\ref{subsec:model} is based on executions.
Here we rephrase it using Kripke models.

\subsection{Kripke frames}


A Kripke frame is defined in terms of a set of  (global) states, together with the following accessibility
relation.  
Two states $u, v \in S$ are  \emph{indistinguishable} by $a$, $u \sim_a v$, 
if and only if the state of process $a$ is the
same in $u$ and in $v$.
 Notice that $u \sim_a v$ defined this way is indeed an equivalence relation.
 
Let $M=\la S, \sim^A \ra$ and $N=\la T,\sim^A \ra$ be two Kripke frames. A \emph{morphism} of Kripke
frame  $M$ to $N$ is a function $f$ from $S$ to $T$ such that for all $u, v \in S$, for all $a \in A$, 
$u \sim_a v$ implies $f(u) \sim_a f(v)$. 
It easy is to check that   morphisms compose.
We call ${\cK}$ the category of Kripke frames, with morphisms of Kripke frames.
 This category enjoys many interesting properties, among
which the fact that cartesian products exist.
Let $M=\la S, \sim^A \ra$ and $N=\la T,\sim^A \ra$ be two Kripke frames, and
define $M\times N=\la U, \sim^A\ra$ as follows~: states $U$ are pairs $u=(s,t)$ of states
$s\in S$ and $t\in T$ and the accessibility relation is defined as
$(s,t) \sim_a (s',t')$ if and only if $s\sim_a s'$ and $t \sim_a t'$.

\begin{lemma}
\label{lem:Kprod}
 The product of Kripke
frames  is the cartesian product in the categorical sense, coming with projections $\pi_M : M\times N
\rightarrow M$ and $\pi_N : M\times N \rightarrow N$, which are morphisms of Kripke frames.
\end{lemma}


%


\subsection{Distributed computability in terms of Kripke frames}
\label{sec:tasksolvability}

Fix a {model of distributed computation} with 
 $N$- step schedules $\mathsf{Sch}$, and the corresponding
initial states $\Gz$ as defined at the end of Section~\ref{sec:statActs}.
The \emph{initial Kripke frame} is  $\cI=\la \Gz, \sim^A \ra$. Notice, that no process can distinguish between
two states in $ \Gz$, while the environment always distinguish them.

Given a deterministic protocol $D$, and an integer $N$,
the   \emph{protocol Kripke frame} is $P=\la \Gp,\sim^A \ra$,
where $\Gp$ consists of all the states at the end of all $N$-step executions starting in $\Gz$.
Abusing notation, we use $\cP$ as the function which sends each $x\in\Gz$ to the state $x'\in\Gp$
obtained by executing the schedule $x_e$ defined by the environment in the initial state $x$.
Examples of protocol Kripke frames are in Section~\ref{sec:apps}.

We define a \emph{inputless task} $\cT=\la \Gz,\Gout,\Delta \ra$ in terms of an \emph{output Kripke frame,} $\Gout$
where each element $x\in \Gout$ defines a value to be decided by each process.
 (the environment is not part of $x$).
The relation\footnote{In previous papers, the  problem is a task, and hence $\Delta$ is a carrier map.} 
$\Delta$ is from $\Gz$ to ${\Gout}$.
Let $x\in  \Gz$ and $x'$ any state in $\Delta(x)$.
Then,  in an execution with schedule indicated by $x_e$,
it is valid for the processes to decide   $x'$, i.e., each $i$ decides 
$x'_i$.


Protocol $D$ \emph{solves in $N$ steps} the inputless task $\cT=\la \Gz,\Gout,\Delta \ra$  if
there exists  a  morphism $f$ from $P=\la {\Gp}, \sim^A \ra$ to  $O=\la {\Gout},\sim^A \ra$ such that the composition of  $P$
and $f$ belongs to $\Delta$, i.e.,  $f(P(x))\in\Delta(x)$.

Let us discuss this definition. 
If $f(u)=x$  then indeed each process $i$ can decide (operationally, in its program)
the value  $x_i$, because the value is a function only of its local state: 
 if $u \sim^i v$  then $f(u)\sim^i f(v)$, namely, $f(u)_i= f(v)_i$.
Second, these decisions are respecting the task specification, because if we consider an initial state $s_0$, then the execution
starting in $s_0$ ends in state  $s= \cP(s_0)$, which is then mapped to a state $t=f(s)$, with $t\in \Delta(s_0)$.
Finally, if such a  morphism  $f$ does not exist, then it is impossible to solve the task in $N$ steps by a deterministic
protocol $D$ in model $M$, because any such protocol would actually be defining the required morphism.

\begin{figure}
\begin{center}
\begin{tikzpicture}
  \node (s) {$\Gfull$};
  \node (xy) [below=2 of s] {$\Gp$};
  \node (x) [left=of xy] {$\Gz$};
    \node (xz) [below=of xy] {$\Gout$};
  \draw[->] (x) to node [sloped, above] {$F$} (s);
  \draw[->, dashed, right] (s) to node {$f$} (xy);
  \draw[<-] (xy) to node [above] {$P$} (x);
    \draw[->, dashed, right] (xy) to node {$f'$} (xz);
      \draw[->] (x) to node [sloped, above] {$\Delta$} (xz);
\end{tikzpicture}
\end{center}
\caption{Full-information protocol is initial among protocols}
\label{fig:initial}
\end{figure}
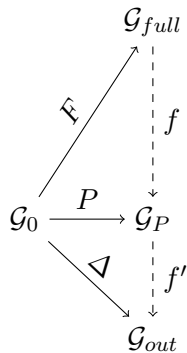
Consider any protocol $D$, and its protocol Kripke frame after $N$ steps, $P=\la {\Gp}, \sim^A \ra$,
and let $F$ be the \emph{full information protocol}, where each time a process writes, it writes its local state,
and when it reads, it concatenates the value read to its local state.
Let $F=\la {\Gfull}, \sim^A \ra$ be its $N$ step protocol Kripke frame. As the Lemma below
formalizes it, and as observed in a slightly different formal context in \cite{Sam}, the ($N$ step)
full information protocol
is \emph{initial} in the wide sub-category of Kripke frames consisting of ($N$ step) protocols (see Figure \ref{fig:initial})~: 

\begin{lemma}
\label{lem:fullInfo}
Protocol~$F=\la {\Gfull}, \sim^A \ra$
solves $\la \Gz,\Gp,P \ra$. 
\end{lemma}
For the proof, notice that
the functions $F$ and $P$ are determined by the $N$-step schedules.  Thus, given an initial state $x$ with schedule $sc$,
and $y=P(x)$, define $f$ by $f(F(x))=P(x)$. To show that $f$ is a morphism, notice that if $u \sim^i v$ for $u,v\in \Gfull$
then $f(u) \sim^i f(v)$ because if $i$ does not distinguish between $u,v$ in a full information protocol, it certainly
does not distinguish in any other protocol.

\begin{corollary}
\label{lem:fullInfoPower}
If  $\cT=\la \Gz,\Gout,\Delta \ra$ is solvable in $N$ steps by a protocol $\cP=\la {\Gp}, \sim^A \ra$ with morphism $f'$, then it is solvable by 
the $N$-step full information protocol with morphism $f\circ f'$. 
\end{corollary}

\subsection{Dynamic epistemic logic}
 Additional background on dynamic epistemic logic
is in Appendix~\ref{app:dynEpLog}.

Let $\I{AP}$ be a countable set of \emph{atomic propositions}.
If $X \subseteq \Lit(\I{AP})$, then $X$ is {\em $\I{AP}$-maximal}  iff
$\forall$ $p \in \I{AP}$, either $p \in X$ or $\neg p \in X$.
Assume a set $A = \{a_0, a_1, \ldots a_n\}$ of $n+1$ agents and 
a countable set $\I{AP}$ of propositional variables.
An \emph{epistemic model}  $M = \la S,\sim^A,L^{\I{AP}} \ra$ consists of a Kripke frame
 $\la S, \sim^A\ra$  and a function
$L^{\I{AP}} : S \to 2^{\Lit(\I{AP})}$ 
 such that 
$\forall s \in S$, $L(s)$ is consistent and $\I{AP}$-maximal.
We will often suppress explicit reference to the sets $\I{AP}$ and $A$,
and denote an epistemic model as $M = \la S,\sim,L \ra$.
The \emph{knowledge} $K_a$ of an agent $a$ with respect to
a state $s$ is the set of formulas which are true in all states
$a$-accessible from $s$.
%

We can organize Kripke models as a category, by defining Kripke model morphisms.
Let $M=\la S, \sim^A,L^{\I{AP}} \ra$ and $N=\la T,\sim^A,L^{\I{AP}} \ra$ be two Kripke
models. A morphism of Kripke models is a
morphism $f$ of the underlying Kripke frames of $M$ and $N$ such that 
$  L^{\I{AP}}(f(s)) \subseteq  L^{\I{AP}}(s) $ for all states $s$ in $S$
(although in the sequel we will have for all our morphisms 
$L^{\I{AP}}(f(s))=L^{\I{AP}}(s)$). 


The following simple lemma from~\cite{ericSergioDEL1-2017} says that morphisms can only ``lose knowledge'' (whereas it
is well-known that $p$-morphisms preserve knowledge \cite{sep-dynamic-epistemic})~: 

\begin{lemma}
\label{prop:prop1}
Consider now two Kripke models  
$M'=\la S',\sim'^A , L'\ra$ and 
$M = \la S,\sim^A, L\ra$, and a 
morphism $f$ from 
$M$ to 
$M'$. 
Then  for every agent $a\in A$, 
for all states $s\in M$, $M',f(s) \models K_a \phi \Rightarrow M,s \models K_a \phi$. 
\end{lemma}


\begin{proof}
Recall that~:
$$M,s \models K_a \varphi \mbox{ iff } \mbox{for all } s' \in S : s
\sim_a s' \mbox{ implies } M,s' \models \varphi$$

Consider $t \sim_a s$~:  
$f(t) \sim_a f(s)$ and as $M',f(s) \models K_a \phi$, 
by definition of the semantics of $K_a$ that we recapped above, we know
that $M',f(t) \models \phi$. Therefore $\phi \in L^{AP}(f(t))$, and by definition
of morphisms of Kripke models, $\phi \in L^{AP}(f(t)) \subseteq L^{AP}(t)$. 
So $M,t \models \phi$ and $M,s \models K_a \phi$. 

\end{proof}


We now turn our attention to information change. Recall that in DEL
an \emph{action model} is a structure $\Msf = \mbox{$\la \Ssf,\sim,\pre \ra$}$,
where $\Ssf$ is a domain of \emph{action points}, such that for
each $a \in A$, $\sim_a$ is an equivalence relation on $\Ssf$, and
$\pre : \Ssf \to \Lcal$ is a precondition function that assigns a
\emph{precondition} $\pre(\ssf)$ to each $\ssf \in \Ssf$.
Each action can be thought of as an announcement made by the environment,
which is not necessarily public, in the sense that
not all system agents receive these announcements.

Let $M=\la S, \sim^A,L^{\I{AP}} \ra$ be a Kripke model and $A = \la T,\sim,\pre \ra$
 be an action model. In the
 \emph{product update model}
 \footnote{Usually pointed 
Kripke models and action models are used, but we do not need them here.}
  $M[A]= \la S\times T, \sim^A,L^{\I{AP}} \ra$, 
each world of  $M[A]$ is a pair $(s,t)$ where $s\in S,t\in T$, such that $\pre(t)$ holds in $s$.
Then, $(s,t)\sim_a (s',t')$ if and only if $s\sim_a s'$ 
and $t\sim_a t'$. The valuation of $p$ at a pair $(s,t)$ is as it was at $s$.
Therefore the underlying Kripke frame of the product update model
$M[A]$ is the cartesian product of the underlying Kripke frames of $M$ and
of $A$. 


\subsection{Action models for distributed computing}
\label{sec:logicTasks}

We extend the task formalism in terms of Kripke frames from Section~\ref{sec:tasksolvability} to define
an epistemic model. 

\subsubsection{Action models for inputless tasks}
Let $\Gz$ be the input Kripke frame of all possible $N$-step schedules in a given model.
Consider a {inputless task}, $\cT=\la \Gz,\Gout,\Delta \ra$, where 
  $\Gout$ is the output Kripke frame, and $\Delta$ is a relation from $\Gz$ to $\Gout$.

Define  atomic propositions that state  what the id of a process is, and others that define what the schedule is.
Then, the \emph{input model}  is $I= \la \Gz,\sim^A,L^{\I{AP}} \ra$, where
 $\la \Gz, \sim^A\ra$ is the input  Kripke frame,
  and the function
$L^{\I{AP}} : S \to 2^{\Lit(\I{AP})}$ 
 states that the id of process $i$ is $i$, and that the environment is in state $sc$, for some $N$-step schedule
 of the model.
 

The action model for  $\cT$ is
$\mbox{$\la \Ssf,\sim,\pre \ra$}$, defined as follows.

The {action points} in $\Ssf$, are identified with the states of $\Gout$, so action point $ap=\la d_0,\ldots,d_n \ra$ is 
interpreted as ``agent $i$ decides value $d_i$'', with precondition that is true in every  input state $u$ such
that $ap\in \Delta(u)$. Thus, the action model has as Kripke frame precisely $O=\la \Gout,\sim^A \ra$.

The  \emph{output model} is obtained by the product update of the {input model}   $ \la \Gz,\sim^A,L^{\I{AP}} \ra$
and the action model $\cT = \mbox{$\la \Ssf,\sim,\pre \ra$}$. Each state in the product can be represented
by $\la sch,ap\ra$, where $sch$ identifies an initial state in $\Gz$.
These states are labeled with the same atomic propositions as the initial state of $sch$.
Also, two states  satisfy 
$\la sch,ap\ra\sim_i\la sch',ap'\ra$
 when process $i$ decides the same value, $ap_i = ap'_i$.

\subsubsection{Action models for protocols}
Consider an {input model}   $I= \la \Gz,\sim^A,L^{\I{AP}} \ra$, for the 
 $N$-step schedules in a given model.


The \emph{$N$-step action model} for protocol $D$ is $A = \la R,\sim,\pre \ra$, where each action point in $R$ corresponds
to an $N$-step schedule of the distributed computing model, and $u\sim_i v$ whenever process $i$
ends up in the same state after schedule $u$ and after schedule $v$, running $D$. 
The precondition is the identity, stating that $\pre (sc)$ is true in the initial state of $\Gz$ with schedule $sc$.
Note that when $D$ is the full-information protocol, if $u\sim_i v$,
then in \emph{every} protocol, process $i$ cannot distinguish schedule $u$ from schedule $v$,
and we get the (sort of canonical, or initial as we noticed earlier) \emph{$N$-step action model of the distributed computing model.}

Given an input  model $I= \la \Gz,\sim^A,L^{\I{AP}} \ra$ 
and an action model $A = \la R,\sim,\pre \ra$ for protocol $D$,
we get 
the product update, \emph{protocol Kripke model} $I[A]= \la \Gz\times R, \sim^A,L^{\I{AP}} \ra$,
which is decorated with the atomic propositions from $I$, and its
 Kripke frame is the  protocol Kripke frame.
A state  in the product can be represented  by  $ac$, an $N$-step schedule of the model,
and has the same atomic propositions as the initial state corresponding to $ac$.
Thus, as far as problem solvability is concerned, one can view the action model
as acting on a single initial state of the processes, and replicating it once for each possible $N$-step schedule,
and copying the atomic propositions, and having the same indistinguishability relation as in the action model.

We can thus define formally a model of computation as a set of infinite schedules, together with an indistinguishability
relation defined on finite prefixes of those schedules, induced by the full-information protocol. 
Namely, two things are needed
to completely define a model. First,  a set of schedules, which specify properties such as,   that at most $t$ processes crash, or
some partial synchrony assumption. A schedule is  just a sequence of sets
of agents to be scheduled, and hence the same set of schedules  can be used for different models.
Thus, remarkably, the effect of different models is captured  by $\sim$, 
the relation specifying when a process cannot possibly  distinguish
between two schedules.  The relation  is implementable, in the sense that the full information protocol
 achieves precisely the  relation.
 Then, for each integer $N$ there is a corresponding set
of initial states  $\Gz$, and an $N$-step action model $A = \la R,\sim,\pre \ra$ characterizing the reachable 
states  
in the distributed computing model.
Other protocols are defined by a coarsening of the relation $\sim$ of the full-information protocol, 
see Lemma~\ref{lem:fullInfo}.


In the  following theorem, we lift the task solvability definition of Section~\ref{sec:tasksolvability}
from the category of Kripke frames to the category of Kripke models.
For short, we say that a morphism $h$  from the protocol Kripke 
model $I[A]$ to the output model $I[T]$ \emph{respects} $\Delta$ if 
for each state $x\in  \Gz$ with scheduler $x_e$,  and the corresponding state $x'\in I[A]$,
the state $y=h(x')$, is such that $y\in\Delta(x)$. 

\begin{theorem}
\label{thm:Kripketasksolv2}
Let $\cT=\la \Gz,\Gout,\Delta \ra$ be an inputless task,
and consider the corresponding input model $I$, and protocol action model $A$ 
and task action model $T$.
Then task solvability 
is equivalent to the existence of a Kripke  morphism $h$ from the protocol Kripke 
model $I[A]$ to the output model $I[T]$, that respects $\Delta$.
\end{theorem}

We  can only improve knowledge from
$I$ to $I[A]$ (the protocol should improve knowledge of the processes
about the execution  through 
communication), and by Lemma~\ref{prop:prop1},   the task is solvable if and only if 
enough knowledge is gained, such that
there is $h$ associating states of
$I[A]$ to states of  $I[T]$, in a way that
any formula $\phi$ and  process $i$ (at any state $y$), 
$I[T],y \models K_i \phi$ then $I[A],s \models K_i \phi$, in any $x$ with $h(x)=y$. That is,
at each state $x$, there is
at least as much knowledge
of the processes at $x$ than at $h(x)$.

\subsection{Kripke models and topology}
\label{sec:categories}

We  briefly describe the equivalence of categories between Kripke
models and  simplicial complex models from~\cite{ericSergioDEL1-2017}, and thus
transport  the semantics of distributed systems from
 (proper, i.e. models in which no distinct state are indistinguishable by all agents) Kripke models to simplicial complex models.
 This exposes the fact that  knowledge
exhibits topological invariants, and that the possibility of solving
an inputless task depends on such invariants, by bringing in results  developed elsewhere e.g.~\cite{HerlihyKR:2013}. 

A \emph{$n$-dimensional complex} $C$  is a family of subsets of a set $S$, called \emph{simplexes}, such that
for all $X \in C$, $Y \subseteq X$ implies $Y \in C$, and the largest simplexes have all exactly $n+1$ elements; maximal elements are the \emph{facets},  and because they all have the same size we say the complex is \emph{pure}.
The elements of a simplex are called \emph{vertices}. Each vertex of a simplex will be \emph{colored} with a different process id (through some map $l$).
We will consider \emph{chromatic  simplicial maps}, $f: C \rightarrow D$,  a function that
maps the vertices of a complex $C$ to the vertices of a complex $D$, preserving ids, and such that a if
a set of vertices $s$ are a simplex of $C$ then $f(s)$ is a simplex of $D$.
Let $p{\mathcal CS}$ be the category of pure chromatic $n$-complexes. (see Appendix~\ref{app:combTop} for notations and additional details).   

A Kripke frame for the $n+1$ processes is equivalent to the $n$-dimensional chromatic complex,
where states of the Kripke frame are uniquely identified to facets of the complex. Two facets $s,s'$ share
a vertex with id $i$ iff the associated states in the frame satisfy $s\sim_i s'$.
Kripke morphisms correspond to chromatic simplicial maps.

\begin{lemma}
\label{thm:equiv}
Let $A$ be a finite set and $p{\mathcal CS}_A$ (resp. ${\mathcal K}_A$) be 
the full subcategory of pure chromatic simplicial complexes with colors in $A$ (resp. the full 
subcategory of proper Kripke frames with agent set $A$).
$p{\mathcal CS}_A$ and ${\mathcal K}_A$ are equivalent categories.
\end{lemma}





The equivalence of categories we have between  pure chromatic simplicial
complexes and Kripke frames can be extended to hold between Kripke models 
models) 
and these pure chromatic simplicial complexes, where facets are decorated with
AP-maximal literals. Simplicial maps $f$ are then extended to map these literals 
associated to facets $X$ to 
the same literals, associated to facet $f(X)$. In fact, a more common way to describe actual
states in this combinatorial topological approach is by decorating vertices of simplicial complexes
by local states of processes. 

Thus we formally define a simplicial model as a triple $(C, l, v)$ where 
$(C,l)$ is a pure chromatic simplicial set,  
and $v: S \rightarrow \wp(\I{AP})$ an assignment of subsets of the set of literals of $\I{AP}$  
for each state $s \in S$ such that for all facets $f=(s_0,\ldots,s_n) \in C$, 
$\bigcup\limits_{i=0}^n v(s_i)$ (that we denote as $v(f)$ by an abuse of notation) is $AP$-maximal.

Lemma~\ref{thm:equiv} extends to the following, in a straightforward
manner~: 

\begin{theorem}
\label{thm:equiv2}
Let $A$ be the set processes, and ${\mathcal SM}_{A,\cG}$ (resp. ${\mathcal KM}_{A,\I{AP}}$) be 
the full subcategory of pure simplicial models with colors in $A$ and states in 
$\cG$ (resp. the full 
subcategory of proper Kripke models with agent set $A$ and atomic propositions in $\I{AP}$)
and suppose we have a complete interpretation $\semb . \seme$ 
of propositions in $\I{AP}$ in 
$\cup_{i \in \{0,\ldots,n\}} L_i$.
Then ${\mathcal SM}_{A,\cG}$ and ${\mathcal KM}_{A,\cG}$ are equivalent categories.
\end{theorem}

By this equivalence between (proper) Kripke models and simplicial models, 
the  approach to task solvability  described in the previous section
 can easily be rephrased in terms of simplicial
 maps, as in \cite{HerlihyKR:2013}.
We can interpret the task solvability Theorem~\ref{thm:Kripketasksolv2}
in purely topological terms, the interest being
that some topological invariants will prevent us from finding a map $h$ as above, showing
impossibility of the corresponding task specification. 
In case of wait-free read/write memory models, we know that
$P(I)$ corresponds to some subdivision of $I$, hence $\pi_I: P(I) \rightarrow I$ is 
a weak homotopy equivalence, this restricts a lot what task specifications
$\Delta$ can be solved, as also exemplified below.


\section{Examples and applications}
 \label{sec:apps}
 
 \subsection{Action model for  $\mathsf{IIS}$}
 \label{sec:theIIScase}
Consider the iterated immediate snapshot model $\mathsf{IIS}$ of Section~\ref{subsec:model}, 
obtained by composing the one-round $\mathsf{IS}$ model $N$ times.
   Processes communicate through a sequence of arrays, $\mathsf{mem}_1$, $\mathsf{mem}_2\ldots,\mathsf{mem}_N$,
 executing an immediate snapshot $\mathsf{IS}()$ operation on each $\mathsf{mem}_r$.

This model is represented by schedules of the environment as follows.
  A \emph{block action} $sc_i$ is an ordered partition $[s_0,\ldots,s_k]$ of the set of ids $A=\{ 0,\ldots, n\}$, consisting
of \emph{concurrency classes}, $s_i$. The concurrency classes $s_i$ are non-empty, disjoints subsets of $A$, whose union is $A$, representing that all processes in $s_i$ are concurrently scheduled to write, and then they are all concurrently scheduled to read.
Thus, $0\leq k\leq n$. When $k=n$ processes will be scheduled  sequential (processes take immediate snapshots one after the other), and when $k=0$  fully concurrent (they all execute an immediate snapshot concurrently).
Let us denote by $L^{sc}_e$ the set of all possible block actions.

The initial states~$\Gz$ for one round, $N=1$, are as follows. The initial states of all processes are identical, except that the initial state of process $i$
contains its id $i$. The initial states of the environment are in a 1 to 1 correspondence to all possible block actions, $L^{sc}_e$; for each block action  $[s_0,\ldots,s_k]$  there
is an initial state of the environment. In addition, the environment's state encodes that the shared memory is initially empty.
Thus, each state in  $\Gz$ can be denoted as $([s_0,\ldots,s_k], q_0,q_1,\ldots,q_n)$,
where $[s_0,\ldots,s_k]$ is a state of the environment and $q_i$ is the initial state of process $i$. 

The {$1$-step action model} for   protocol $D$ is $A = \la R,\sim,\pre \ra$, 
has an  action point in $R$ for each block action  $[s_0,\ldots,s_k]$. For the accessibility relation $\sim$,
define $view_i(act)$,  $act=[s_0,\ldots,s_k]$, to be the set of ids that are scheduled before or together with process $i$,
namely;  $view_i(act)=s_0\cup\cdots\cup s_j$, where $i\in  s_j$.
For two block actions  $act=[s_0,\ldots,s_k]$, $act'=[s'_0,\ldots,s'_{k'}]$, it holds that $act\sim_i act'$
iff 
$view_i(act)=view_i(act')$.
Indeed, in \emph{every} protocol $D$, process  $i$ cannot tell if the schedule applied is $act$ or $act'$, in both it reads
all values written by processes in $view_i(act)=view_i(act')$. 
Furthermore, in  \emph{every} protocol $D$  (even not full-information)
if $view_i(act)\neq view_i(act')$, then it is not the case that $act\sim_i act'$, assuming a process always writes
something to the shared memory. Thus we have the following.

\begin{lemma}
For every protocol $D$, the $1$-step action model  is $A = \la R,\sim,\pre \ra$, 
where $R$ consists of all block actions  $[s_0,\ldots,s_k]$ and $\sim$ on two block actions $act,act'$
is defined by
$act\sim_i act'$
iff 
$view_i(act)=view_i(act')$.
\end{lemma}

Furthermore, the $N$-step action model for the iterated immediate snapshot model $\mathsf{IIS}$
is easily obtained because action models compose~\cite{DEL:2007}, and 
because the iterated model uses a fresh memory in each round\footnote{
We confuse the notation of the actual number of steps a process executes, and the number of rounds
in the IIS model, to avoid further notation}.
(For the non-iterated version of the model, the accessibility relation $\sim$ is more complicated, see~\cite{AttiyaR2002}.)
Notice however, that the $\mathsf{IIS}$ model is represented by the action composition only for the
full-information protocol, because in another protocol $D$, even if a process $i$ was able to distinguish
between two schedules in the first round,  $D$  might not announce it to the shared memory in the second iteration.

\begin{theorem}
The $N$-step action model  for the full-information protocol $D$ is the composition of the $1$-step action model, $N$-times.
\end{theorem}

The set of all $N$-step runs~$\cR$ of protocol~$D$ in the  $\mathsf{IIS}$ model
that start in initial states~$\Gz$ can be obtained by applying schedules of the following form. 
Every execution in~$\cR$  is  of the form
$x\odot sc_1\odot sc_2\odot\cdots\odot sc_R$ where $x\in\Gz$ and
 $sc_i$ is a {block scheduling action}
for every integer $1\leq i\leq R$. 
To apply the block action $sc_1=[s_0,\ldots,s_k]$ to initial state $([s_0,\ldots,s_k], q_0,q_1,\ldots,q_n)$, and obtain state $x\odot [s_0,\ldots,s_k]$, the 
environment schedules the processes in the following order, to execute their read and write operations on $\mathsf{mem}_{1}$.
 It first schedules the processes in $s_0$ to execute their
write operations, and then it schedules them to execute their read operations (the specific order among writes
is immaterial, and so is the case for the reads). Then the environment repeats the same for the processes in $s_1$,
scheduling first the writes and then the reads, and so on, for each subsequent concurrency class $s_i$.
See Figure~\ref{fig-3scheduls}.
Notice that the read and write operations of block action $sc_i$ are applied to memory $\mathsf{mem}_{i}$.
See Figure~\ref{fig-allScheds}.

Consider the composition of the block actions $sc_1\odot sc_2\odot\cdots\odot sc_R$.
Let $IIS_R$ denote the set of all such composition of block actions. That is, 
$$
sc_1\odot sc_2\odot\cdots\odot sc_R \in IIS_R
$$ 
if and only if each $sc_i$ is an ordered partition of $A$, $[s_0,\ldots,s_k]$.
Then, any execution of~$\cR$ of protocol~$D$ in the  $\mathsf{IIS}$ model can be obtained by applying
a composed block action of $IIS_R$ to an initial state in $\Gz$.
In other words, $\cR$ can be seen as the product of $\Gz$ and $IIS_R$.

\begin{figure}
\begin{center}
\includegraphics[scale=0.3]{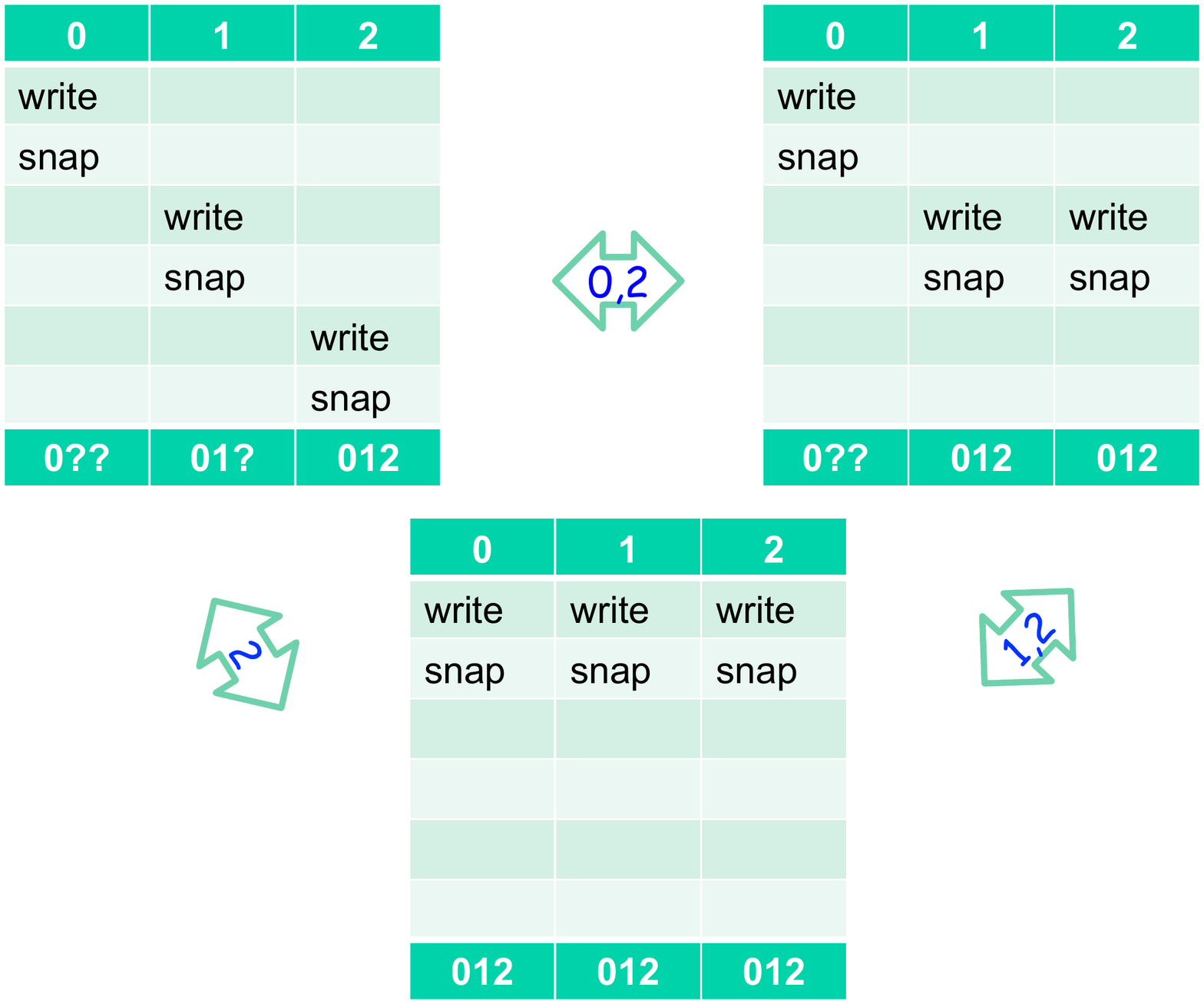}
\end{center}
\caption{Schedules $[\{0\},\{1\},\{2\}]$, $[\{0\},\{1,2\}]$, $[\{0,1,2\}]$, the arrows are labeled with processes
that do not distinguish between the schedules.
}
\label{fig-3scheduls}
\end{figure}

\begin{figure}
\begin{center}
\includegraphics[scale=0.3]{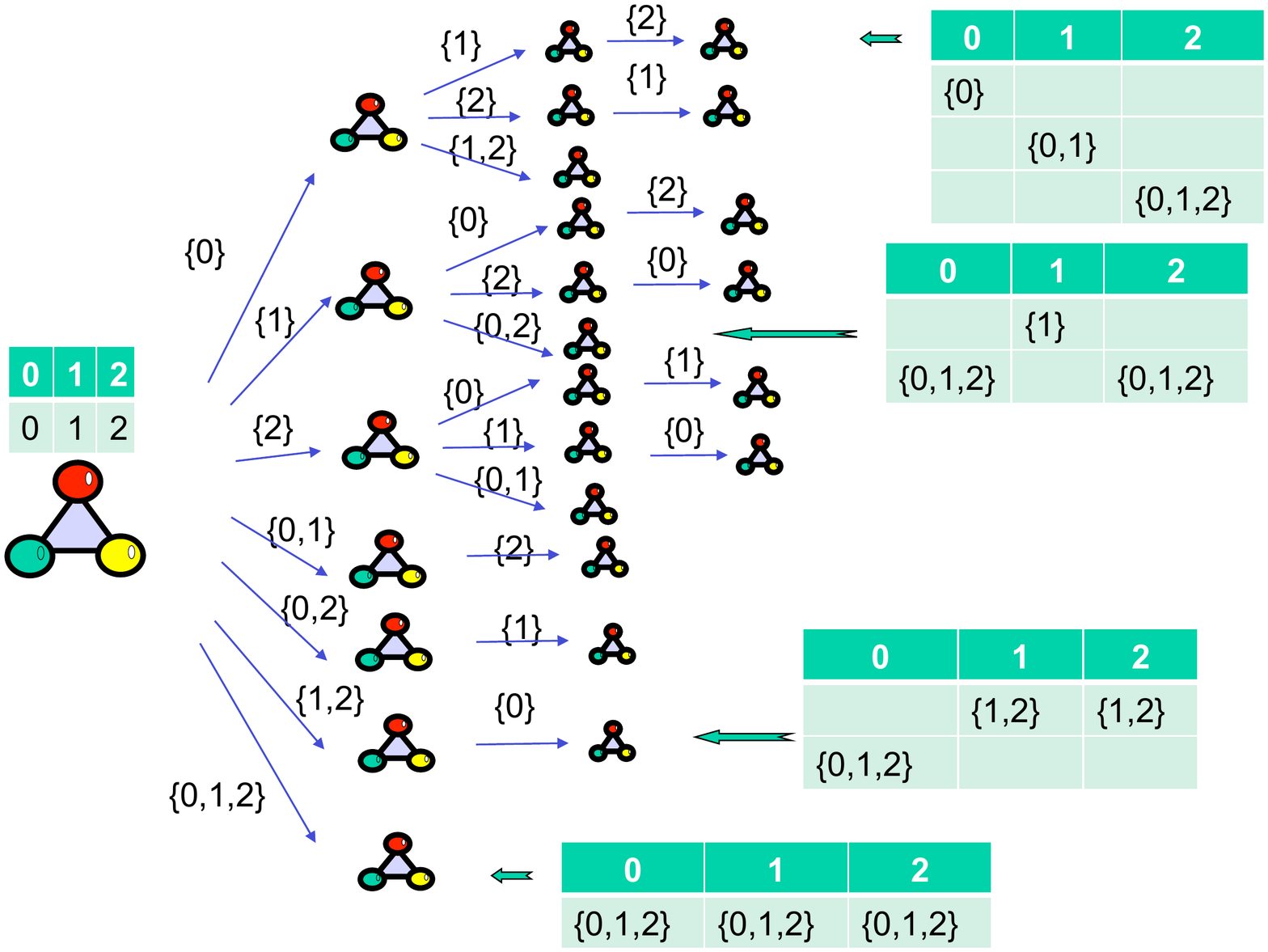}
\end{center}
\caption{All block schedules for 3 processes. The bottom row of each table contains the views of the processes at the end
of the schedule. 
}
\label{fig-allScheds}
\end{figure}

In Figure~\ref{fig-subSched-dim2-1}  the protocol complex after one round  and after two
rounds are illustrated, with several examples of schedules. 
Each vertex of a simplex has a color
that represents one of the agents.
The input complex is not
depicted, it consists of a single $2$-dimensional simplex and all its faces (a triangle).
Thus, it corresponds to a single initial state $x$.
The complex on the lower left corner is a chromatic subdivision of the input triangle. It
represents the protocol complex after one round, where each one of three processes 
(in the figure called $p,q,r$) are scheduled to execute first one write operation to $mem_1$,
and then one read of all the registers in $mem_1$. 
 
 The green simplex in the lower left complex
is obtained by applying the schedule $\{p\}\{qr\}$ to the initial state $x$.
First $p$ writes to its component of $mem_1$ then it reads all three components.
Then $q$ and $r$ are scheduled to concurrently write to their components of $mem_1$,
and finally $q$ and $r$ are scheduled to concurrently read $mem_1$.
Similarly, the schedule $\{p\}\{q\}\{r\}$ schedules first $p$ (its write followed by its reads),
then it schedules $q$, and finally it schedules $r$.
Notice that  $\{p\}\{qr\}$ and $\{p\}\{q\}\{r\}$ are points of the action model,
and they are indistinguishable to both $p$ and to $r$, which is why the green and the yellow
triangle share an edge labeled with the colors of $p$ and $r$. Furthermore,
these two schedules are indistinguishable whenever applied to two initial states
indistinguishable to both $p$ and $r$.
The schedule where all three write concurrently and then all three do their reads
concurrently consists of a single concurrency class $\{pqr\}$, and yields the triangle
at the center.

Now, consider the complex at the top right of the figure. 
The triangle at the center of the green part of the complex
is obtained by applying to $x$ the schedule $\{p\}\{qr\}\odot \{pqr\}$,
where the processes first access $mem_1$ following  $\{p\}\{qr\}$
and then $mem_2$ following $ \{pqr\}$.
Consider the other green triangle in the second round complex,
obtained by the schedule
$\{p\}\{qr\}\odot \{pr\}\{q\}$. Notice that now the
 action points are indistinguishable to $q$ only, that is,
 $\{p\}\{qr\}\odot \{pqr\}\sim_q \{p\}\{qr\}\odot \{pr\}\{q\}$, and indeed the two green
 triangles share only a white vertex.

An interesting example is the yellow triangle at the corner of the 2nd round complex,
obtained by the schedule
$\{p\}\{qr\}\{q\}\odot \{p\}\{r\}\{q\}$, namely, processes are scheduled
sequentially to access first $mem_1$ and then again sequentially $mem_2$.

Consider the green and yellow triangles in the two round protocol, obtained
through the 2nd round schedule  $\{pr\}\{q\}$. The green one
is obtained by $\{p\}\{qr\}\odot \{pr\}\{q\}$ while the yellow one by
$\{p\}\{q\}\{r\}\odot \{pr\}\{q\}$. It is illustrating to consider the one round chromatic
subdivision complex as an input complex, and then apply the 1-round action model
to it, to obtain the complex at the top right.

We would like to stress that the same schedule may have different semantics
in different models.
In this iterated model, 
$\{p\}\{r\}\{q\}\odot \{p\}\{r\}\{q\}$ is equivalent to the schedule
$\{p\}\{p\}\{r\} \{r\}\{q\}\{q\}$. In this last schedule, $p$ writes and read $mem_1$
and then $mem_2$, and only then $r$ writes and reads $mem_1$ and then $mem_2$,
and finally $q$ does the same. But the state obtained is exactly the same as
in the schedule $\{p\}\{r\}\{q\}\odot \{p\}\{r\}\{q\}$ (the three processes end up in the same
local states in both schedules).
In contrast, in the non-iterated version of the model, where there is only
one shared array $mem$, the two schedules produce different states.
Actually, in the non-iterated IS model, the action points
$\{p\}\{r\}\{q\}\odot \{p\}\{r\}\{q\}$ and $\{p\}\{p\}\{r\} \{r\}\{q\}\{q\}$
are distinguishable to the three processes!

Recall that in the non-iterated IS model executions are organized in concurrency classes,
where in each on, a set of processes is scheduled to first write (to their corresponding registers) in $mem$
and then read $mem$ (all the registers). The action model for $N$  has one point for each such schedule,
where each process executes the same number of operations, $N$. The indistinguishability
relation $\sim_i$ of when process $i$ does not distinguish between the two schedules
is characterized in~\cite{AttiyaR2002}.

\begin{figure}
\begin{center}
\includegraphics[scale=0.25]{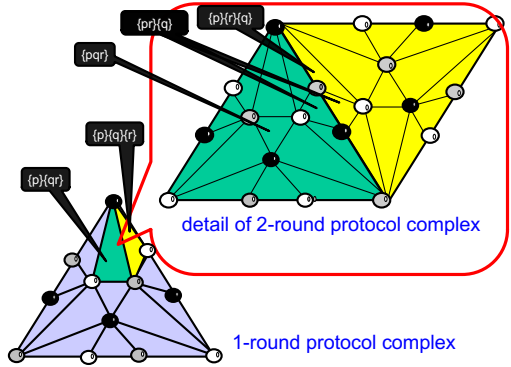}
\end{center}
\caption{Schedules for the iterated model (from~\cite{HerlihyKR:2013}).
A black vertex is associated to $p$, a white one to $q$ and a grey one to $r$.
}
\label{fig-subSched-dim2-1}
\end{figure}

\subsection{Solving a task in  $\mathsf{IIS}$}
 \label{sec:aTask}
 
 Examples  of inputless problems have been mentioned in the introduction. 
 Here we bring together the pieces of the framework developed in the paper,
 by studying the solvability of a simple inputless task. To simplify the presentation,
 we develop only the case of 3 processes, because all the essential elements appear already
 here.
 
Consider here  the $2$-test\&set inputless task for three processes, 
specifying that in any execution at most two process should output $1$ and the others should output $0$,
and such that if one (or two) process terminates its execution without seeing any other processes, it should 
 output $1$ (or if two finish without seeing the third, they both should output $1$).
\begin{figure}
\begin{center}
\includegraphics[scale=0.33]{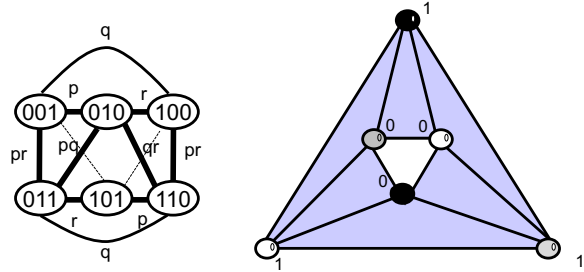}
\end{center}
\caption{Output Kripke frame and corresponding output complex for the 2-test\&set task.
A black vertex is associated to $p$, a white one to $q$ and a grey one to $r$.
}
\label{fig-testSetTwo}
\end{figure}
The {inputless task} in terms of  Kripke frames is $\cT=\la \Gz,\Gout,\Delta \ra$, where $\Gz$
defines $N$-step schedules. The
{output Kripke frame,} $\Gout$, has a corresponding $2$-dimensional chromatic complex, as
is illustrated in Figure~\ref{fig-testSetTwo} (processes are denoted $p,q,r$ to distinguish them from output values $0,1$).
The input/output relation $\Delta$ is specified,  as follows. Any input state $s$ in $\Gz$ with a schedule
  $act$, where $i$ never reads a value by another process,  $\Delta$ requires $i$ to decide $1$. 
  Similarly,  any state $s$ in $\Gz$ with a schedule
  $act$, where $i$ and $j$ never read a value by the other process,  $\Delta$ requires $i$ and $j$ to decide $1$,
  and the other process to decide $0$.
  In all other cases, $\Delta$ allows any outputs where are most two processes decide $1$. 
 
 \begin{theorem}
 \label{th:testSetTwo}
 The 2-test\&set task is not  solvable in the IIS model.
 \end{theorem}
 
 \begin{proof}
 In the proof we will be moving freely (and abusing notation) between the category
 of Kripke models and that of complexes, by Theorem~\ref{thm:equiv2}.
 
Let   $I= \la \Gz,\sim^A,L^{\I{AP}} \ra$ be the $N$-step {input model} for the IIS model.
Thus, each state $s\in\Gz$ corresponds to an $N$-step schedule of the IIS distributed computing model.

 Let    $T=\mbox{$\la \Gout,\sim,\pre \ra$}$, be the action model for the 2-test\&set task,
where the {action points} are identified  with binary decisions $ap=\la d_0,d_1,d_2 \ra$, not all equal. 
An action point $act$ has precondition that is true in every  input state $u$ such
that $ap\in \Delta(u)$, where $\Delta$ is the input/output relation of the 2-test\&set task.
The action model has as Kripke frame precisely $O=\la \Gout,\sim^A \ra$ illustrated
in Figure~\ref{fig-testSetTwo}.

 Consider a full-information protocol $D$ for $N$-steps. 
 First notice that we can assume without loss of generality that $D$ is full information 
 (e.g. see Lemma~\ref{lem:fullInfo}).
 Now, let   $A = \la R,\sim,\pre \ra$ be the {$N$-step action model} for protocol $D$ is
  where each action point in $R$ corresponds
to an $N$-step schedule of the distributed computing model.
The  {protocol Kripke model} $I[A]= \la \Gz\times R, \sim^A,L^{\I{AP}} \ra$
is obtained by the product update of the action model $A$ and the input model $I$.
The corresponding $2$-dimensional chromatic complex for $N=2$ is illustrated in 
Figure~\ref{fig-subSched-dim2-1}, which is indeed, in general, 
an iterated chromatic subdivision (e.g. see~\cite{HerlihyKR:2013}).

If the task is solvable by the protocol $D$, by Theorem~\ref{thm:Kripketasksolv2}, there 
exists morphism $h$ from the protocol Kripke 
model $I[A]$ to the output model $I[T]$, that respects $\Delta$.

Let $\pi$ be the projection morphism from (the underlying Kripke frame of) $I[T]$ 
(which is $I\times T$) to the output complex $\Gout$.
Then  $\pi\circ h$ is a chromatic simplicial  map from $I[A]$ to $\Gout$.
However, this map cannot cannot exist, because $I[A]$ is a subdivision of a simplex,
and $\pi\circ h$ sends the boundary of $I[A]$ to the boundary of $\Gout$.
This is a contradiction, because the boundary of $I[A]$ is contractible to a point within $I[A]$, 
while the  boundary of $\Gout$ is not.
 \end{proof}

\section{Conclusions}
We have developed a framework to give a formal semantics in terms of 
dynamic epistemic logic (DEL) to distributed computing models where processes communicate
by reading and writing shared registers. We showed how the model of computation itself
can be represented by an action model, $A$, consisting of the possible schedules of the model,
and a relation defining when a process does not distinguish between two schedules.
The DEL product update $I[A]$ of the action model with the initial model $I$ represents knowledge gained by
the processes in the model. 
We showed how to model an inputless task also by an action model, $T$. The product update $I[T]$
represents knowledge that processes should be able to gain to solve the task, formally expressed by the
existence of a certain 
 morphism from $I[A]$ to  $I[T]$.
 Finally, by moving freely between  the category of simplicial complex models, and the equivalent  category
 of Kripke models, we bring benefits back and forth, between DEL and the topological approach to distributed
 computing. 
 
 The framework refines the work of our companion paper~\cite{ericSergioDEL1-2017},
 that included the formalisation of knowledge change only with respect to the inputs of the processes,
 and serve well for input/output tasks.
 Both settings can be incorporated into one setting, but when studying only input/output it would
 overly complicated. Many open questions remain, these papers are only the beginning of a
 longer term project to study fault-tolerant distributed computing from a dynamic epistemic logic perspective.
 For example, we work only with simplicial complexes where all facets have dimension $n$, and
 this is indeed sufficient to study any task in models where failures are not detectable in finite time.
 The framework can be extended to arbitrary complexes, to deal well with synchronous
 systems where processes can crash, this will be the subject of a forthcoming article.

 \subparagraph*{Acknowledgements.}
  S. Rajsbaum would like to acknowledge UNAM PAPIIT~~under Grant
  No.IN109917
    and~the Ecole Polytechnique for financial support through the 2016-2017 Visiting Scholar Program.

\newpage

\appendix




\section{Dynamic epistemic logic background}
\label{app:dynEpLog}
We adhere to the notation of~\cite{DEL:2007}.
Let $\I{AP}$ be a countable set of \emph{atomic propositions} (i.e.,
\emph{propositional variables}).
The set of {\em literals} over $\I{AP}$ is
$\Lit(\I{AP})= \I{AP} \cup \{\neg p \mid p \in \I{AP}\}$.
The {\em complement} of a literal $p$ is defined by
$\overline{p} = \neg p$ and $\overline{\neg p} = p$, $\forall p \in \I{AP}$.
If $X \subseteq \Lit(\I{AP})$, then
$\overline{X}= \{ \overline\ell \mid \ell \in X \}$; 
$X$ is {\em consistent} iff 
$\forall$ $\ell \in X$, $\overline{\ell} \notin X$;
and $X$ is {\em $\I{AP}$-maximal}  iff
$\forall$ $p \in \I{AP}$, either $p \in X$ or $\neg p \in X$.

\begin{definition}[syntax] 
Let $\I{AP}$ be a countable set of propositional variables and $A$ a set of agents.
The language $\mathcal{L}_K(A,\I{AP})$ (or just $\mathcal{L}_K$ when the context makes it clear) 
is generated by the following BNF
grammar:
\[
\varphi ::= p \mid \neg\varphi \mid (\varphi \land \varphi) \mid
K_a\varphi
\]
\end{definition}

\begin{definition}[Semantics of formulas]
Consider an epistemic state $(M,s)$ with $M=\la S,\sim,V \ra$ a Kripke model, $s \in S$ and
$\varphi, \psi \in \Lcal_K(A,\I{AP})$.
The satisfaction relation, determining when a formula is true in that 
epistemic state, is defined as:

\begin{tabular}{lrl}
$M,s \models p$ & iff & $p \in L(s)$\\
$M,s \models \neg \varphi$ & iff & $M,s \not\models \varphi$\\
$M,s \models \varphi \wedge \psi$ & iff & $M,s \models \varphi
\mbox{ and } M,s \models \psi$\\
$M,s \models K_a \varphi$ & iff & $\mbox{for all } s' \in S : s
\sim_a s' \Rightarrow M,s' \models \varphi$\\
\end{tabular}
\end{definition}


Hence, 
an agent $a$ is said to know an assertion 
in a state $(M,s)$
iff that assertion of true in all the states it considers possible,
given $s$. 
Therefore, the \emph{knowledge} $K_a$ of an agent $a$ with respect to
a state $s$ is the set of formulas which are true in all states
$a$-accessible from $s$.

\begin{definition}[Language of action model logic]
\emph{
We define the language of action model logic $\Lcal_{\I{KC}\otimes}(A,\I{AP})$ for a set of agents
$A$ and propositional variables $AP$ 
as the set of 
formulas $\varphi \in \Lcal^{\RM{stat}}_{\I{KC}\otimes}(A,\I{AP})$ and of  
actions $\alpha \in \Lcal^{\RM{act}}_{\I{KC}\otimes}(A,\I{AP})$, defined through the following grammar~: 
\begin{align*}
\varphi & ::= p \mid (\neg\varphi) \mid (\varphi \wedge \varphi) \mid
K_a \varphi \mid [\alpha]\varphi\\
\alpha & ::= (\Msf,\ssf) 
\end{align*}
where $p \in P$, $a \in A$.
}
\end{definition}

\begin{definition}[Semantics of formulas and actions]
Consider an epistemic state $(M,s)$ with $M=\la S,\sim,L \ra$ a Kripke model,
an action model $\Msf = \la \Ssf,\sim,\pre \ra$, and
$\varphi \in \Lcal^{\RM{stat}}_{\I{KC}\otimes}(A,P)$, 
$\alpha=(\Msf,\ssf) \in \Lcal^{\RM{act}}_{\I{KC}\otimes}(A,P)$. The satisfaction relation between formulas
and epistemic states is given below, as a set of inductive rules~: 

\begin{tabular}{lrl}
$M,s \models p$ & iff & $p \in L(s)$\\
$M,s \models \neg \varphi$ & iff & $M,s \not\models \varphi$\\
$M,s \models \varphi \wedge \psi$ & iff & $M,s \models \varphi
\mbox{ and } M,s \models \psi$\\
$M,s \models K_a \varphi$ & iff & $\mbox{for all }
s' \in S : $ \\ & & $s
\sim_a s' \mbox{ implies } M,s' \models \varphi$\\
$M,s \models [(\Msf,\ssf)]\varphi$ & iff & 
$M,s \models \pre(\ssf) \mbox{ implies }$\\
& & $(M \otimes \Msf,(s,\ssf)) \models \varphi$ \\
\end{tabular}


The \emph{restricted modal product} $M \otimes \Msf = \la S',\sim',L'
\ra$ is defined as:

\begin{tabular}{lrl}
$S'$ & $=$ & $\{(s,\ssf) \mid s \in S, \ssf \in \Ssf,$ \\
& & \ \ \ $\mbox{ and } M,s
  \models \pre(\ssf)\}$\\
$(s,\ssf) \sim_a' (t,\tsf)$ & iff & $s \sim_a t \mbox{ and } \ssf \sim_a
  \tsf$\\
$p \in L'(s,\ssf)$ & iff & $p \in L(s)$
\end{tabular}

\end{definition}


\section{Combinatorial topology background}
\label{app:combTop}
For a textbook covering combinatorial topology notions see~\cite{kozlov:2007}.


\begin{definition} [Simplicial complex]
A simplicial complex $C$ is a family of non-empty finite subsets of a set $S$ such that
for all $X \in C$, $Y \subseteq X$ implies $Y \in C$ ($C$ is downwards closed).
\end{definition}

Elements of $S$ (identified with singletons) 
are called vertices, elements of $C$ of greater cardinality are called faces. The dimension of a face $X \in C$, $\dim X$, is the cardinality of $X$ minus one. 
The maximal faces of $C$ (i.e. faces that are not subsets of any other face) are called facets. 
The dimension of a simplicial complex is the maximal dimension of its faces. Pure simplicial complexes are simplicial complexes such that the maximal faces are all of the same
dimension. 

\begin{definition} [Simplicial maps] 
Let $C$ and $D$ be two simplicial complexes. A simplicial map $f: C \rightarrow D$ is a function that
maps the vertices of $C$ to the vertices of $D$ such that for all faces $X$ of $C$, $f(C)$ (the image set
on the subset of vertices $C$) is a face of $D$. 
\end{definition}

Now, we can define pure simplicial maps respecting facets. 
Finally, we will associate colors to each vertex of simplicial complexes, representing, as in~\cite{HerlihyKR:2013}, the names of the different processes
involved in a protocol. We also define chromatic simplicial maps as the simplicial
maps respecting colors. This can actually be seen as a slice category constructed
out of the pure simplicial category.

\end{document}